\documentclass[preprint,11pt,a4paper]{elsarticle}

\usepackage{xcolor}
\usepackage{threeparttable}
\usepackage{mathrsfs}
\usepackage{graphicx}
\usepackage{subfigure}
\usepackage{booktabs}
\usepackage{url}
\usepackage{multirow}
\usepackage{setspace}
\usepackage{url}
\usepackage{epstopdf}
\usepackage{amssymb,amsmath,amsthm,mathabx}
\usepackage[ruled]{algorithm2e}

\usepackage{colortbl,booktabs}
\usepackage{xcolor}
\usepackage[toc,page,title,titletoc,header]{appendix}
\usepackage{bm}
\usepackage{float}
\usepackage{lineno}
\usepackage{tikz}
\usetikzlibrary{decorations.markings}
\usepackage{adjustbox}
\usepackage{enumitem}

\usepackage[bindingoffset=0.2in,%
            left=1in,right=1in,top=1in,bottom=1in,%
            footskip=.25in]{geometry}

\newtheorem{theorem}{Theorem}
\newtheorem{definition}{Definition}
\newtheorem{lemma}{Lemma}

\newtheorem{corollary}{Corollary}

\graphicspath{{figures/}}

\newcommand{\J}{{\mathcal{J}}}
\newcommand{\F}{{\mathcal{F}}}
\newcommand{\x}{{\bm{x}}}

\global\long\def\LP{\text{\normalfont LP}}
\global\long\def\MLP{\text{\normalfont Multi-LP}}
\global\long\def\NewMLP{\text{\normalfont New-Multi-LP}}

\global\long\def\OPT{\text{\normalfont OPT}}
\global\long\def\SOL{\text{\normalfont SOL}}

\newcommand{\WT}[1]{{\color{black}{#1}}}

\usepackage{soul}
\setstcolor{orange}

\allowdisplaybreaks
\graphicspath{{figures/}}

\journal{Theoretical Computer Science}

\begin{document}

\begin{frontmatter}



\title{A polynomial-time approximation scheme for parallel two-stage flowshops under makespan constraint}

\author[a]{Weitian Tong$^*$}
\ead{wtong.research@gmail.com}
\author[a]{Yao Xu}
\ead{yxu@georgiasouthern.edu}
\author[b]{Huili Zhang}
\ead{zhang.huilims@xjtu.edu.cn}

\address[a]{Department of Computer Science, Georgia Southern University, Statesboro, GA 30458, USA}
\address[b]{School of Management, Xi’an Jiaotong University; and State Key Lab for Manufacturing Systems Engineering, Xi'an 710049, China}
\cortext[*]{Corresponding author}

\begin{abstract}
As a hybrid of the Parallel Two-stage Flowshop problem and the Multiple Knapsack problem, 
we investigate the scheduling of parallel two-stage flowshops under makespan constraint,
which was motivated by applications in cloud computing and introduced by Chen et al. \cite{CHG21} recently. 
A set of two-stage jobs are selected and scheduled on parallel two-stage flowshops 
to achieve the maximum total profit while maintaining the given makespan constraint. 
We give a positive answer to an open question about its approximability proposed by Chen et al. \cite{CHG21}.
More specifically, 
based on guessing strategies and rounding techniques for linear programs,
we present a polynomial time approximation scheme (PTAS) 
for the case when the number of flowshops is a fixed constant. 
\end{abstract}

\end{frontmatter}

%
%

\section{Introduction}\label{sec_introduction}


We study the scheduling of parallel two-stage flowshops under makespan constraint. 
Suppose there are $m$ identical two-stage flowshops $\F = \{F_j, j \in \{1, 2, \ldots, m\} \}$ and 
a set of $n$ jobs $\J = \{J_i, i \in \{1, 2, \ldots, n\}\}$.
Once assigning a job $J_i$ to some flowshop $F_j$, job $J_i$ needs to be processed non-preemptively on $F_j$
and the first and second operation of $J_i$ has a workload of $a_i$ and $b_i$ respectively. 
Meanwhile, finishing processing $J_i$ brings a profit $p_i$.
The objective is to identify the most profitable subset of jobs, denoted by $\J^{selected}$, such that 
the minimum makespan of $\J^{selected}$, i.e., the completion time of the last job, is bounded by 1. 
\footnote{The makespan constraint can be any real value. Without loss of generality, the boundary value is set to 1, 
as we are always able to scale the job sizes.}
In other words, we aim at packing two-stage jobs to parallel flowshops 
in order to achieve the maximum total profit. 
Therefore, we also name the studied problem as the \emph{(m, 2)-Flowshop-Packing} problem.

The study of parallel two-stage flowshops under makespan constraint was recently introduced by Chen et al. \cite{CHG21}
and was motivated by applications from cloud computing \cite{WCW19a}. 
Receiving a request from a client for a specific resource, 
a server on the cloud will read certain data from disks and then transfer the data back to the client. 
Thus, a server can be regarded as a two-stage flowshop and 
a request from clients can be treated as a two-stage job 
consisting of a disk-reading operation and a network-transition operation. 
The cloud service provider aims to maximize the profit under time constraint.

For a maximization optimization problem $\Pi$ and an approximation algorithm ${\cal A}$, 
the approximation ratio of ${\cal A}$ is defined as
$\max_{I \in \mathcal{I}} \{ {\cal A}(I)/{\mbox{OPT}(I)} \}$, 
where $\mathcal{I}$ is the set of instances. 
A maximization problem admits a \emph{polynomial-time approximation scheme} (PTAS) 
if there is a family of algorithms $\{ \mathcal{A}_{\epsilon}, \epsilon > 0\}$ such that 
each algorithm $\mathcal{A}_{\epsilon}$ has an approximation ratio of $1 - \epsilon$ for any $\epsilon > 0$
and its time complexity is $O(n^{f(1/\epsilon)})$ for instance size $n$ and some function $f$.
We say a PTAS is an \emph{efficient polynomial-time approximation scheme} (EPTAS) if 
the running time of $\mathcal{A}_{\epsilon}$ is in the form of $f(1/\epsilon)\cdot n^{O(1)}$;
and 
we say a PTAS is a \emph{fully polynomial-time approximation scheme} (FPTAS)
if the running time of $\mathcal{A}_{\epsilon}$ is a polynomial in $\frac{1}{\epsilon}$, 
that is, 
$(1/\epsilon)^{O(1)} \cdot n^{O(1)}$.


It is not hard to observe that the classic Knapsack problem \cite{GJ79} is a special case of (1,2)-Flowshop-Packing 
by setting $b_i = 0$ for all jobs. 
Similarly, the (m,2)-Flowshop-Packing problem generalizes the Multiple Knapsack problem with an uniform capacity,
which admits no FPTAS for the two-knapsack case unless P = NP \cite{CK05}. 
When $m$ is part of the input, Multiple Knapsack with an uniform capacity is strongly NP-hard \cite{ZG86}.
Therefore, 
(1,2)-Flowshop-Packing inherits the NP-hardness from the Knapsack problem.
(m,2)-Flowshop-Packing admits no FPTAS even for $m=2$
and becomes strongly NP-hard when $m$ is part of the input.
When $m$ is a constant, 
Chen et al. \cite{CHG21} designed a pseudo-polynomial time algorithm for (m,2)-Flowshop-Packing,
implies the weakly NP-hardness of (m,2)-Flowshop-Packing.

A special case with only one flowshop was initially explored by Dawande et al. \cite{DGR06},
who proposed an approximation algorithm with a ratio of $1/3 - \epsilon$ by 
utilizing the FPTAS for the Knapsack problem \cite{IK75}.
Recently, Chen et al. \cite{CHG21} introduced the multiple-flowshops version 
and study this problem systematically. 
When the number of flowshops is part of the input, 
they first presented an computationally efficient $1/4$-approximation algorithm and then 
improved the ratio to $1/3 - \epsilon$ by exploring the connection 
between the Multiple Knapsack problem and the (m, 2)-Flowshop-Packing problem.
%
When the number of flowshops is a fixed constant, 
they integrated approximation techniques for the classical Knapsack problem and the Parallel Machine Scheduling problem
and designed a $(1/2 - \epsilon)$-approximation algorithm,
which also improves the approximability result by Dawande et al. \cite{DGR06}.
Chen et al. \cite{CHG21} then proposed an open question that 
whether the problem has polynomial-time approximation schemes 
particularly when the number $m$ of flowshops is a fixed constant.
We give a positive answer for the case when $m$ is a fixed constant.

\begin{table}[H]
\begin{center}
\caption{A summary of known results for the (m,2)-Flowshop-Packing problem. 
        Our contributions are highlighted with bold font.}\label{table_results}
\begin{tabular}{|p{0.3\textwidth}|l|l|}
\hline
Problem & Hardness & Approximability  \\ \hline
\multirow{3}{*}{(1,2)-Flowshop-Packing} & \multirow{3}{*}{Weakly NP-hard \cite{DGR06}} 
                 & $1/3-\epsilon$ \cite{DGR06} \\ \cline{3-3}
~& ~             & $1/2-\epsilon$ \cite{CHG21} \\ \cline{3-3}
~& ~             & \textbf{$1-\epsilon$ (Our PTAS)} \\ \hline
\multirow{2}{0.3\textwidth}{(m,2)-Flowshop-Packing ($m \geq 2$ is fixed)} & \multirow{2}{*}{Weakly NP-hard \cite{CHG21}, No FPTAS \cite{CK05}}
                 & $1/2-\epsilon$ \cite{CHG21} \\ \cline{3-3}
~& ~             & \textbf{$1-\epsilon$ (Our PTAS)} \\ \hline
(m,2)-Flowshop-Packing ($m \geq 2$ is part of the input) & Strongly NP-hard  \cite{ZG86}
                 & $1/3-\epsilon$ \cite{CHG21} \\
\hline
\end{tabular}
\end{center}
\end{table}

Observing that scheduling jobs on one two-stage flowshop to achieve the minimum makespan 
can be resolved in polynomial time by the famous Johnson's Algorithm \cite{Joh54},
we roughly solve the studied problem by two steps:
(1) identifying a high-profit job subset;
(2) scheduling the selected jobs on each flowshop with Johnson's Algorithm.
The first step is more involved as 
it should be on the alert for jobs with large workload
so that the second step will not violate the makespan constraint.
Our main contribution is a PTAS for the (m,2)-Flowshop-Packing problem 
by combining guessing strategies and rounding techniques in linear programming. 
It is worth remarking that \emph{guessing} a quantity in polynomial time
means identifying a polynomial number of values for this quantity in polynomial time.

Our first idea is to guess in polynomial time the most profitable subset of size $\frac{1+\epsilon}{\epsilon}$ in the optimal solution,
where $\epsilon > 0$ is chosen such that $\frac{1+\epsilon}{\epsilon}$ is an integer.
Consequently, the remaining jobs that need to be selected are relatively cheaper. 
Based on the quantitive formula (refer to Eq.(\ref{eq_Johnson_makespan})) to estimate the makespan 
\WT{when a job sequence is under a permutation schedule},
our second idea is to design a linear program to select the cheap jobs fractionally. 
It turns out that this linear program guarantees a strong property.
That is, except for at most a constant number of cheap jobs, 
most cheap jobs can be identified integrally via the optimal solution to this linear program.

To illustrate our ideas more clearly, 
we start with resolving the (1,2)-Flowshop-Packing problem by designing a PTAS for it.
Then we extend our ideas to the general case, i.e., the (m,2)-Flowshop-Packing problem, and
present another PTAS when the number of flowshops is a fixed constant.
Our results significantly improve the approximation ratios obtained 
by Chen et al. \cite{CHG21} and Dawande et al. \cite{DGR06}.
Please refer to Table \ref{table_results} for existing results and our contributions.

\paragraph{Organization}
We review the most related works in Section \ref{sec_relatedwork}.
Section \ref{sec_preliminary} introduces some essential notions and terminologies.
We present and analyze the PTASes for the (1,2)-Flowshop-Packing and (m,2)-Flowshop-Packing problems
in Section \ref{sec_ptas1} and Section \ref{sec_ptas2} respectively. 
In Section \ref{sec_conclusion}, we make a conclusion 
and discuss future research directions.

\section{Related Work}\label{sec_relatedwork}
The (m, 2)-Flowshop-Packing problem aims at maximizing the total profit by 
selecting and scheduling jobs on parallel two-stage flowshops with makespan constraint.
It is closely related to the Parallel Two-stage Flowshop Scheduling problem \cite{Kov85} and 
the Multiple Knapsack problem \cite{Kel99}. 

%
\paragraph*{The Parallel Two-stage Flowshop Scheduling problem}
Removing the makespan constraint and replacing the objective with makespan minimization
result in another scheduling problem, named \emph{Parallel Two-stage Flowshop Scheduling},
which has attracted quite a few attentions \cite{Kov85,HKA96,VE00,ZV12,DTL17,DJL20,WCW19a,WCW19b,WCW20a,WCW20b}. 
The parallel two-stage flowshop scheduling problem is NP-hard when $m\geq 2$ 
and becomes strongly NP-hard if $m$ is part of the input \cite{DJL20}. 
Kovalyov \cite{Kov85}, Dong et al. \cite{DTL17}, and Wu et al. \cite{WCW19a}
designed FPTAS independently for the parallel two-stage flowshop scheduling problem
when $m$ is a fixed constant. 
For the case where $m$ is part of the input, 
Dong et al. \cite{DJL20} eventually presented a PTAS
after a few explorations of its approximability by Wu et al. \cite{WCW20b,WCW20a}.
It is worth noting that for the parallel multi-stage flowshop scheduling problem, 
Tong et al. \cite{TMG18} proposed a PTAS when both the number of flowshops and the number of stages are constant. 
It is still open whether a PTAS exists for the parallel multi-stage flowshop scheduling problem 
when $m$ is part of the input.

\paragraph*{The Multiple Knapsack problem}
The (m, 2)-Flowshop-Packing problem reduces the Multiple Knapsack problem 
by setting the workload of each second-stage operation to zero. 
The Multiple Knapsack problem is weakly NP-hard for any constant number of knapsacks,
as it admits a pseudo-polynomial dynamic programming \cite{GRR19}, which generalizes the classic one for the Knapsack problem.
An FPTAS exists for the classic Knapsack problem \cite{IK75}
while Multiple Knapsack with at least two knapsacks admits no FPTASes \cite{CK05}. 
When the number of knapsacks is a part of the input, 
Multiple Knapsack becomes strongly NP-hard \cite{ZG86} even for the special case 
where all knapsacks have an uniform capacity. 
Kellerer \cite{Kel99} developed a PTAS for the special case of the Multiple Knapsack problem 
where all knapsacks have same capacity. 
Chekuri and Khanna \cite{CK05} proposed a PTAS for the general Multiple Knapsack problem 
with running time $O(n^{\log (1/\epsilon) / \epsilon^8})$.
They first cleverly rounded instances in a more structured one 
that has logarithmic profit values and size values;
then grouped knapsacks by their capacities in a structured way;
finally sequentially packed the most profitable items, large-size items, 
and the remaining cheap and small-size items with different strategies. 
Jansen \cite{Jan10,Jan12} improved the time efficiency by 
presenting an EPTAS with parameterized running time $O\left(2^{\frac{1}{\epsilon}\log^4 (1/\epsilon)} + poly(n) \right)$. 
\WT{All the above PTASes and EPTASes} were designed for the case where $m$ is part of the input.

Both Chen et al. \cite{CHG21} and  Dawande et al. \cite{DGR06} 
utilized the connection between the Multiple Knapsack problem and the (m, 2)-Flowshop-Packing problem
to design approximation algorithms for the (m, 2)-Flowshop-Packing problem. 
In particular, they observed that a (Multiple) Knapsack instance can be constructed from an (m, 2)-Flowshop-Packing instance
by linearly combining the workload at two stages for each job. 
Our PTAS does not utilize such connections. 
Instead, we rely on the structure of the optimal scheduling (with respect to the makespan) guaranteed by Johnson's Algorithm
on each flowshop.

\section{Preliminary}\label{sec_preliminary}

In this section, we define the studied problem formally 
and introduce essential notations and important concepts.

A standard two-stage flowshop contains one machine at every stage and
a two-stage job has two operations.
Once a job is assigned to a flowshop, 
its two operations are processed non-preemptively on the two sequential machines in the flowshop, respectively.
In particular, the second operation cannot start processing until the first operation has been completed. 
The \emph{makespan} is defined as the completion time of the last job.
Denote the set $\{1, 2, \ldots, n\}$ by $[n]$ for any positive integer $n$.
The formal definition of the scheduling of parallel two-stage flowshops under makespan constraint 
(or the (m,2)-Flowshop-Packing problem) 
is shown in Definition \ref{def_problem}.

\begin{definition}[The (m,2)-Flowshop-Packing problem]\label{def_problem}
Given a set of $m$ identical two-stage flowshops $\F = \{F_j, j \in [m] \}$ 
and a set of $n$ jobs $\J = \{J_i = (a_i, b_i, p_i), i \in [n]\}$,
where $a_i$ and $b_i$ denote the processing time or workload of $J_i$ on the first and second stage respectively 
and $p_i$ represents the profit by processing $J_i$, 
the goal is to identify a subset of jobs, denoted by $\J^{selected}$, and to schedule them on the given flowshops 
such that
the total profit is maximized while the makespan is limited within 1.
\end{definition}

Assume $a_i + b_i \leq 1$ holds for all $i \in [n]$ without loss of generality. 
We abuse notations $a$, $b$, $p$ as linear summation functions over a job set. 
That is, for any subset of jobs $\J'$, 
$a(\J') = \sum_{J_i \in \J'} a_i$, $b(\J') = \sum_{J_i \in \J'} b_i$, and $p(\J') = \sum_{J_i \in \J'} p_i$.
Suppose $\J^*$ is the job set chosen in an optimal solution.
Let $\OPT = p(\J^*) = \sum_{J_i \in \J^*} p_i$ denote the optimal total profit.

The classic Two-Stage Flowshop problem minimizes the makespan of a single two-stage flowshop
and can be solved optimally by the well-known Johnson's Algorithm \cite{Joh54}.
More precisely, we have the following Theorem. 

\begin{theorem}\cite{Joh54,BEP19}\label{thm_Johnson}
For the two-stage flowshop problem, there exists a permutation schedule which returns the minimum makespan. 
In this optimal schedule, job $J_i$ precedes job $J_{i'}$ if 
\WT{$\min\{a_i, b_{i'}\} \leq \min\{a_{i'}, b_{i}\} $}.
\end{theorem}

We say a job sequence is under Johnson's order if 
\WT{$\min\{a_i, b_{i'}\} \leq \min\{a_{i'}, b_{i}\} $} implies job $J_i$ precedes job $J_{i'}$.
The order of any two jobs is independent of all the other jobs under Johnson's order.
Then Corollary \ref{corollary} follows immediately.

\begin{corollary}\label{corollary}
Given a job sequence under Johnson's order, any subsequence of this job sequence is under Johnson's order.
\end{corollary}

\WT{
Without loss of generality, all job sequences under consideration in the rest of this paper 
are assumed in Johnson's order by default.  
}

\WT{Given any permutation schedule for two-stage flowshop jobs,
it is well-known that its makespan can be computed as 
\begin{equation}\label{eq_Johnson_makespan}
\max_{s \in [n]} \left\{ \sum_{i=1}^s a_i + \sum_{i=s}^n b_i \right\}.
\end{equation}
}
We say a job is a \emph{critical} job if it has the 
the subscript $\arg \max_{s \in [n]} \left\{ \sum_{i=1}^s a_i + \sum_{i=s}^n b_i \right\}$.
The tie is broken by taking the job with the smallest subscript index.

\section{PTAS for the (1,2)-Flowshop-Packing problem}\label{sec_ptas1}

In this section, a PTAS is presented for the (1,2)-Flowshop-Packing problem,
in which there is only one two-stage flowshop. 
Solving (1,2)-Flowshop-Packing can be naturally partitioned into two steps:
selecting a subset of jobs and then finding a feasible schedule on the flowshop for the chosen jobs. 
As Johnson's Algorithm is able to schedule jobs to achieve the minimum makespan,
we focus on the first step.
For any small positive constant $\epsilon \in (0, 1)$,
the rough idea is to guess a constant number $\frac{1+\epsilon}{\epsilon}$ of the most profitable jobs in the optimal job set $\J^*$ 
and then select cheaper jobs carefully such that 
the chosen job set $\J^{selected}$ has its overall profit at least $(1-\epsilon)\cdot \OPT$ 
while the minimum makespan is at most 1.

In sequel, we assume $|\J^*| > \frac{1+\epsilon}{\epsilon}$, as otherwise the optimal solution can be found
by exhausting $n^{O(1/\epsilon)}$ subsets, each having a size of at most $\frac{1+\epsilon}{\epsilon}$.
We guess $\frac{1+\epsilon}{\epsilon}$ most profitable jobs in $\J^*$
and denote this subset of jobs by $\J^{profitable}$.
%

\begin{lemma}\label{lemma_cheap_jobs}
The cheapest job in $\J^{profitable}$ has a profit less than $\epsilon \cdot \OPT$.
Then, for any job in $\J^* \backslash \J^{profitable}$, its profit is less than $\epsilon \cdot \OPT$.
\end{lemma}

\begin{proof}
Let $p_{\min} = \min_{ J \in \J^{profitable} }  p(J) $. 
Assume $p_{\min} \geq \epsilon \cdot \OPT$,
the total profit of jobs in $\J^{profitable}$ is at least
$$ p(\J^{profitable}) \geq \frac{1+\epsilon}{\epsilon} \cdot p_{\min} \geq
        \frac{1+\epsilon}{\epsilon} \cdot \epsilon \cdot \OPT 
        = (1 + \epsilon) \cdot \OPT 
        > \OPT,$$
which contradicts the definitions of $\J^{profitable}$ and $\OPT$. 


\end{proof}

As the jobs in $\J^{profitable}$ are expected to be the $\frac{1+\epsilon}{\epsilon}$ most profitable jobs in $\J^*$,
all jobs in $\J \backslash \J^{profitable}$ with  profit greater than $p_{\min}$
can be ignored temporarily. 
Let $\J' = \{ J \mid p(J) \leq p_{\min}, J \in \J \} \cup \J^{profitable}$.

Then we guess the critical job, say $J_s$, in the optimal solution $\J^*$. 
A linear program, denoted by $\LP$, is formulated to assign the cheaper jobs fractionally. 
The variable $x_{i}$ is defined as the fraction of job $J_i$ that is assigned to the given flowshop.
\begin{align*}
\max \quad  & \sum_{i\in [n] } p_i \cdot x_i &  & ~~~~ \LP\\
s.t. \quad  & \sum_{i=1}^s a_i x_i + \sum_{i=s}^n b_i x_i \leq 1  & & \mbox{~~~~(makespan constraint)}\\  
~           & x_i = 1, & J_i \in \J^{profitable} \cup \{J_s \} & \mbox{~~~~(guessed assignments)} \\
~           & x_i = 0, & J_i \in \J \backslash \left( \J' \cup \{J_s \} \right) & \\
~           & x_i \in [0, 1], & J_i \in \J' \backslash \left( \J^{profitable} \cup \{J_s \} \right) &
\end{align*}

\begin{lemma}\label{lemma_lp}
If the guesses for the most profitable job set $\J^{profitable}$ and the critical job $J_s$ are correct,
\begin{enumerate}
\item the linear program $\LP$ has a feasible solution, 
        and $\OPT$ is upper bounded by the optimal objective value of $\LP$;
\item there is a polynomial time algorithm to obtain a job set from the optimal basic feasible solution to $\LP$
        such that the makespan is at most 1 and the total profit is at least $(1 - \epsilon) \cdot \OPT$.
\end{enumerate}

\end{lemma}

\vspace{10pt}
\begin{algorithm}[H]
\caption{PTAS for (1,2)-Flowshop-Packing}\label{ptas1}

\textbf{Input}: any constant $\epsilon \in (0, 1)$ and a (1,2)-Flowshop-Packing instance $\{\F, \J \}$\;
\textbf{Output}: a job subset $\J^{selected} \subseteq \J$
                  with the makespan upper bounded by 1\;

~\\

Let $\mathcal{C} = \left\{\J^{guess} \mid |\J^{guess}| \leq \frac{1+\epsilon}{\epsilon}, \J^{guess} \subseteq \J \right\}$\;

\uIf{$n > \frac{1+\epsilon}{\epsilon}$}{
  $\mathcal{C}_{\epsilon} = \left\{\J^{guess} \mid |\J^{guess}| = \frac{1+\epsilon}{\epsilon}, \J^{guess} \subseteq \J \right\}$\;  
}\Else{
  $\mathcal{C}_{\epsilon} = \emptyset$\;   
}

\For{every $\J^{guess} \in \mathcal{C}_{\epsilon}$}{
    $\J^{profitable} \leftarrow \J^{guess}$\;
    Let $p_{\min} = \min_{ J \in \J^{profitable} }  p(J) $\;
    Let $\J' = \{ J \mid p(J) \leq p_{\min}, J \in \J \} \cup \J^{profitable}$\;
    \For{every $J_s \in \J'$} { 
      Use $\J^{profitable}$, $\J'$, and $J_s$ to construct $\LP$\;
      \If{$\LP$ admits a feasible solution}{
              Let $x$ be an optimal basic feasible solution to $\LP$\;
              $\J^{selected} \leftarrow$ $\left\{ J_i \mid x_i = 1 \right\}$\;
              $\mathcal{C} \leftarrow \mathcal{C} \cup \J^{selected} $\;
      }
    }
  }

Let $profit = -\infty$\;
\For{every $\J^{candidate} \in \mathcal{C}$}{
  Let $\pi$ be the schedule of $\J^{candidate}$ by Johnson's Algorithm\;
  \If{$p(\J^{candidate}) > profit$ and $\pi$'s makespan is at most 1}{
      $\J^{selected} \leftarrow \J^{candidate}$ and $profit \leftarrow p(\J^{candidate})$\;
  }
}
\Return $\J^{selected}$
\end{algorithm}
\vspace{10pt}

\begin{proof}
Suppose $\J^*$ is sorted in Johnson's order.
\WT{The minimum makespan for $\J^*$ can be computed by Eq.(\ref{eq_Johnson_makespan}).}
By Corollary \ref{corollary},
$\{x_i = 1 \mid J_i \in \J^* \} \cup  \{x_i = 0 \mid J_i \not\in \J^* \} $ is a feasible solution to $\LP$.
Then the correctness of the first claim follows immediately.

Due to the Rank Lemma for linear programs (Chapter 2.1 in \cite{LRS11}),
the single non-trivial constraint of $\LP$, i.e., the \emph{makespan constraint}, implies that
at most one $x_i$ is exactly fractional, i.e., $x_i \in (0, 1)$.
Suppose $x$ is an optimal basic feasible solution to $\LP$. 
The job set $\{J_i \mid x_i = 1, i \in [n]\}$ is chosen as $\J^{selected}$.

Let $\OPT_{\LP}$ be the optimal objective value of $\LP$
and $\SOL$ be the total profit of $\J^{selected}$.
Then, 
we have
\begin{eqnarray*}
\OPT_{\LP} & \leq & \max_{J_i \in \J' \backslash \left( \J^{profitable} \cup \{J_s \} \right)} \{p_i\} + \sum_{J_i \in \J^{selected}} p_i \\ 
& \leq & p_{\min} + \SOL \\
& \leq & \epsilon \cdot \OPT + \SOL,
\end{eqnarray*}
where the last inequality holds due to the correct guess of $\J^{profitable}$.

By the first part of this Lemma, 
\begin{equation*}
\SOL \geq \OPT_{\LP} - \epsilon \cdot \OPT\geq 
\OPT - \epsilon \cdot \OPT = 
(1 - \epsilon) \cdot \OPT.
\end{equation*}

This completes the proof for this lemma.
\end{proof}

Putting the above ideas together, our algorithm maintains all candidates of $\J^{selected}$.
Let $\mathcal{C}$ denote this collection. 
To consider the case that $\J^*$ contains at most $\frac{1+\epsilon}{\epsilon}$ jobs, 
$\mathcal{C}$ is initialized by $\left\{\J^{guess} \mid |\J^{guess}| \leq \frac{1+\epsilon}{\epsilon}, \J^{guess} \subseteq \J \right\}$.
For the case $|\J^*| > \frac{1+\epsilon}{\epsilon}$,
we exhaust all possible candidates of $\J^{profitable}$, that is, a job subset of size exactly $\frac{1+\epsilon}{\epsilon}$.
For each candidate of $\J^{profitable}$, we guess a critical job and use $\LP$ to obtain a candidate of $\J^{selected}$.
Finally, among all candidates, we select the one with a feasible schedule and the maximum profit.
Note that some candidates in $\mathcal{C}$ may not admit a feasible schedule satisfying the limited makespan of 1.
A detailed description of our algorithm is provided in Algorithm \ref{ptas1}.

\begin{theorem}\label{thm1}
Algorithm \ref{ptas1} is a PTAS for the (1,2)-Flowshop-Packing problem.
\end{theorem}

\begin{proof}
Based on the previous analysis, we claim that the optimal solution $\J^*$ is contained in the candidate collection $\mathcal{C}$.
Lemma \ref{lemma_lp} implies an approximation ratio of $1-\epsilon$ for any $\epsilon \in (0, 1)$.

LP has $O(n)$ constraints and can be solved in polynomial time via any interior point methods \cite{PW00}. 
$\mathcal{C}_{\epsilon}$ has a cardinality of $O(n^{(1+\epsilon)/\epsilon})$.
Thus, the first for-loop in Algorithm \ref{ptas1} takes $poly(n^{1/\epsilon})$ time. 
As Johnson's algorithm has a time complexity of $O(n\log n)$ and 
there are $O(n^{(1+\epsilon)/\epsilon})$ candidates in $\mathcal{C}$,
the second for-loop in Algorithm \ref{ptas1} also takes $poly(n^{1/\epsilon})$ time. 
Thus, Algorithm \ref{ptas1} is a PTAS.
\end{proof}

\section{PTAS for the (m,2)-Flowshop-Packing problem}\label{sec_ptas2}

In this section, we extend the ideas for the (1,2)-Flowshop-Packing problem to 
present a PTAS for the general case, i.e., (m,2)-Flowshop-Packing, when $m$ is a fixed constant. 

As there are $m$ identical flowshops, 
we not only guess $\frac{1+\epsilon}{\epsilon}$ most profitable jobs in the optimal job set $\J^*$, 
still denoted by $\J^{profitable}$, 
but also guess how the jobs in $\J^{profitable}$ distribute among the $m$ flowshops
and the critical job on each flowshop. 
Then, the cheaper jobs are selected by rounding a basic feasible solution of a linear program,
which is more complicated than the one for the (1,2)-Flowshop-Packing problem.

Suppose $\J^{profitable}$ has been correctly guessed. 
Let $\J^{profitable}_j$ denote the subset of jobs in $\J^{profitable}$ that are scheduled on the flowshop $F_j$, $j \in [m]$.
Define the \emph{distribution} of $\J^{profitable}$ on the $m$ flowshops
as the tuple $(\J^{profitable}_1, \J^{profitable}_2, \ldots, \J^{profitable}_m)$.

\begin{lemma}\label{lemma_distribution}
There are at most $O((m/\epsilon)^{1/\epsilon})$ different distributions of $\J^{profitable}$ on $m$ flowshops. 
\end{lemma}

\begin{proof}
Assume only $k$ flowshops contains jobs from $\J^{profitable}$.
Counting the number of different distributions of $\J^{profitable}$ on $m$ flowshops is equal to 
counting the number of different partitions of $\frac{1+\epsilon}{\epsilon}$ objects into $k$ non-empty subsets,
which can be estimated by the Stirling partition number \cite{RD69}, denoted by $S\left(\frac{1+\epsilon}{\epsilon},k \right)$.

Let $ h = \min\left\{m, \frac{1+\epsilon}{\epsilon} \right\}$. We have $k \leq h$.
Considering all possible $k$ values, the number of different distributions of $\J^{profitable}$ on $m$ flowshops 
is upper bounded as
\begin{eqnarray*}
\sum_{k  = 1}^{h} \binom{m}{k} S\left(\frac{1+\epsilon}{\epsilon}, k \right) 
& \leq & \sum_{k = 1}^{h} \binom{m}{k} \left( \frac{1}{2} \binom{\frac{1+\epsilon}{\epsilon}}{k} k^{\frac{1+\epsilon}{\epsilon} - k} \right) \\
& \leq & \frac{h}{2} \cdot (me)^{h} \cdot \left(\frac{1+\epsilon}{\epsilon} \cdot e\right)^{h} \cdot h^{\frac{1+\epsilon}{\epsilon}}, \\
& \leq & \frac{1}{2} \cdot h^{2/\epsilon} \cdot (2me^2/\epsilon)^{h}, 
\end{eqnarray*}
where the first inequality holds due to the bound $S(n,k) \leq \frac{1}{2} \binom{n}{k} k^{n-k}$ \cite{RD69};
the second inequality follows from the fact that $\binom{n}{k} \leq \frac{n^k}{k!} \leq \left(\frac{n \cdot e}{k}\right)^k$;
the last inequality is because $\epsilon$ is small enough positive number. 
\end{proof}

Taking the same strategy in Section \ref{sec_ptas1}, we define 
$p_{\min} = \min_{ J \in \J^{profitable} }  p(J) $
and $\J' = \{ J \mid p(J) \leq p_{\min}, J \in \J \} \cup \J^{profitable}$.

Now we guess the critical job on each flowshop.
Suppose $J_{s_j}$ is the critical job on $F_j$ in the optimal solution. 
To assign cheaper jobs fractionally, a linear program is formulated and we denote it by $\MLP$.
The variable $x_{ij}$ is defined as the fraction of job $J_i$ that is assigned to the flowshop $F_j$.
\begin{align*}
 \max \quad  & \sum_{i\in [n] } \sum_{j\in [m]} p_i \cdot x_{ij}  & & ~~~~\MLP\\
s.t. \quad  & \sum_{i=1}^{s_j} a_i x_{ij} + \sum_{i=s_j}^n b_i x_{ij} \leq 1,  & \forall j \in [m]  & \mbox{~~~~(makespan constraints)}\\  
~           & \sum_{j\in [m]} x_{ij} \leq 1,  &  J_i \in \J' \backslash \left( \J^{profitable} \cup \{J_{s_j}, j \in [m] \} \right) & \\
~           & x_{ij} = 1,  & J_i \in \J^{profitable}_j \cup \{J_{s_j} \}, j \in [m] & \mbox{~~~~(guessed assignments)} \\
~           & x_{ij} = 0,  & J_i \in \J \backslash \left( \J' \cup \{J_{s_j}, j \in [m]\} \right), j \in [m]  & \\
~           & x_{ij} \geq 0 , &  J_i \in \J' \backslash \left( \J^{profitable} \cup \{J_{s_j}, j \in [m] \} \right) & 
\end{align*}

We introduce a dummy flowshop $F_{m+1}$ and 
set a sufficiently large makespan limit  on $F_{m+1}$, say $a(\J) + b(\J)$.
Due to the makespan constraint, $\MLP$ may leave some jobs in $\J$ unassigned. 
The main purpose of the dummy flowshop is to (fractionally) assign 
these unassigned jobs to $F_{m+1}$. 
We formulate another linear program, denoted by $\NewMLP$,
where every job in $\J$ is assigned. 
Note that the objective of $\NewMLP$ does not count the profit of jobs that are assigned to the dummy flowshop $F_{m+1}$.
\begin{align*}
\max \quad  & \sum_{i\in [n] } \sum_{j\in [m]} p_i \cdot x_{ij}  & ~~~~\NewMLP \\
s.t. \quad  & \sum_{i=1}^{s_j} a_i x_{ij} + \sum_{i=s_j}^n b_i x_{ij} \leq 1,  & \forall j \in [m]\\  
~           & \sum_{i=1}^{s_{m+1}} a_i x_{ij} + \sum_{i=s_{m+1}}^n b_i x_{ij} \leq a(\J) + b(\J) \\  
~           & \sum_{j\in [m+1]} x_{ij} = 1,   & J_i \in \J' \backslash \left( \J^{profitable} \cup \{J_{s_j}, j \in [m+1] \} \right)  \\
~           & x_{ij} = 1,   & J_i \in \J^{profitable}_j \cup \{J_{s_j} \}, j \in [m+1] \\
~           & x_{ij} = 0,  & J_i \in \J \backslash \left( \J' \cup \{J_{s_j}, j \in [m]\} \right), j \in [m]   \\
~           & x_{ij} \geq 0, &  J_i \in \J' \backslash \left( \J^{profitable} \cup \{J_{s_j}, j \in [m+1] \} \right) 
\end{align*}

Denote $N^{unassigned}$ as the number of jobs that are unassigned before constructing $\NewMLP$.
Then
$$
N^{unassigned} = |\J' \backslash \left( \J^{profitable} \cup \{J_{s_j}, j \in [m+1] \} \right)|. 
$$
$\NewMLP$ has $N^{unassigned} + m + 1$ non-trivial constraints, which are shown as the first three sets of constraints in $\NewMLP$.
Clearly, the number of assigned variables $N^{unassigned} \cdot (m+1)$ is considerably larger than the number of non-trivial constraints
when both $m$ and $\epsilon$ are fixed constants. 
For any basic feasible solution $x$ to $\NewMLP$, 
the number of positive variables in $x$ is at most the number of non-trivial constraints.

Denote $N^{one}$ as the number of jobs in $\J' \backslash \left( \J^{profitable} \cup \{J_{s_j}, j \in [m+1] \} \right)$ 
that are assigned via a variable $x_{ij}$ with value 1.
Let $N^{frac}$ denote the remaining job subset of $\J' \backslash \left( \J^{profitable} \cup \{J_{s_j}, j \in [m+1] \} \right)$.
Each job in $N^{frac}$ is assigned via exactly fractional variables.
As the new $\NewMLP$ assigns every job in $\J' \backslash \left( \J^{profitable} \cup \{J_{s_j}, j \in [m+1] \} \right)$,
we have 
\begin{equation*}
N^{one} + N^{frac} = N^{unassigned}.
\end{equation*}

If $J_i$ is assigned via exactly fractional variables, 
there are at least two exactly fractional variables associated with $J_i$. 
Therefore, there are at least $2 \cdot N^{frac}$ exactly fractional variables
and the total number of positive variables are at least $2 \cdot N^{frac} + N^{one}$.

To wrap up, we have
\begin{equation*}
N^{unassigned} + N^{frac} = N^{one} + 2 \cdot N^{frac} \leq N^{unassigned} + m + 1,
\end{equation*}
which implies $N^{frac} \leq m + 1$.

\begin{lemma}\label{lemma_mlp}
If all guesses, including $\J^{profitable}_j$ and $J_{s_j}$, $j \in [m]$, are correct, 
\begin{enumerate}
\item the linear program $\MLP$ has a feasible solution,
        and $\OPT$ is upper bounded by the optimal objective value of $\LP$;
\item there is a polynomial time algorithm to obtain a job set from the optimal basic feasible solution to $\MLP$
        such that the makespan is at most 1 and the total profit is at least $(1 - \epsilon \cdot (m+1)) \OPT$.
\end{enumerate}
\end{lemma}

\begin{proof}
The first part of this lemma is obvious by the construction of $\MLP$.

From the definition of $\NewMLP$, we observe that 
restricting any basic feasible solution of $\NewMLP$ to flowshops $\F = \{F_1, F_2, \ldots, F_m\}$ 
results in a basic feasible solution to $\MLP$.
In addition, the optimal objective of $\NewMLP$ is the same to the optimal objective of $\MLP$.

Suppose $x$ is an optimal basic feasible solution to $\MLP$. 
Then $x$ contains at most $m+1$ fractional variables. 
We select job as $\J^{selected} = \{J_i \mid x_{ij} = 1, ~i \in [n] \mbox{ and } j \in [m]\}$.
Let $\OPT_{\MLP}$ be the optimal objective value of $\MLP$
and $\SOL$ be the total profit of $\J^{selected}$.
Then, we have
\begin{eqnarray*}
\WT{\OPT_{\MLP}} & \leq & (m+1) \cdot \max_{J_i \in \J' \backslash \left( \J^{profitable} \cup \{J_s \} \right)} \{p_i\} + \sum_{J_i \in \J^{selected}} p_i \\
& \leq & (m+1) \cdot p_{\min} + \SOL \\ 
& \leq & \epsilon \cdot (m+1) \cdot \OPT + \SOL, 
\end{eqnarray*}
which implies 
\begin{equation*}
\SOL \geq  (1-\epsilon \cdot (m+1)) \cdot \OPT.
\end{equation*}

\end{proof}

\vspace{10pt}
\begin{algorithm}[H]
\caption{PTAS for (m,2)-Flowshop-Packing}\label{ptas2}

\textbf{Input}: any constant $\epsilon \in (0, 1)$ and a (m,2)-Flowshop-Packing instance $\{\F, \J \}$\;
\textbf{Output}: a job subset $\J^{selected} \subseteq \J$
                  with the makespan upper bounded by 1\;

Let $\mathcal{C} = \left\{\J^{guess} \mid |\J^{guess}| \leq \frac{1+\epsilon}{\epsilon}, \J^{guess} \subseteq \J \right\}$\;

\uIf{$n > \frac{1+\epsilon}{\epsilon}$}{
  $\mathcal{C}_{\epsilon} = \left\{\J^{guess} \mid |\J^{guess}| = \frac{1+\epsilon}{\epsilon}, \J^{guess} \subseteq \J \right\}$\;  
}\Else{
  $\mathcal{C}_{\epsilon} = \emptyset$\;   
}
\For{every $\J^{guess} \in \mathcal{C}_{\epsilon}$}{
   $\J^{profitable} \leftarrow \J^{guess}$\;
    Let $p_{\min} = \min_{ J \in \J^{profitable} }  p(J) $\;
    Let $\J' = \{ J \mid p(J) \leq p_{\min}, J \in \J \} \cup \J^{profitable}$\;
   \For{every distribution $(\J^{profitable}_1, \J^{profitable}_2, \ldots, \J^{profitable}_m)$ of $\J^{profitable}$ }{
          \For{every combination $(J_{s_1}, J_{s_2}, \ldots, J_{s_{m+1}})$, $\WT{J_{s_j} \in \J'}, j \in [m+1]$}{ 
            Construct $\NewMLP$\;
            \If{$\NewMLP$ admits a feasible solution}{
                    Let $x$ be an optimal basic feasible solution to $\NewMLP$\;
                    Restrict $x$ to flowshops $\F$, still denoted by $x$\;
                    $\J^{selected} = \{J_i \mid x_{ij} = 1, ~i \in [n] \mbox{ and } j \in [m]\}$\;
                    $\mathcal{C} \leftarrow \mathcal{C} \cup \J^{selected} $\;
            }
          }
  }
}

Let $profit = -\infty$\;
\For{every $\J^{candidate} \in \mathcal{C}$}{
  Let $\pi$ be the schedule of $\J^{candidate}$ by Johnson's Algorithm\;
  \If{$p(\J^{candidate}) > profit$ and $\pi$'s makespan is at most 1}{
      $\J^{selected} \leftarrow \J^{candidate}$ and $profit \leftarrow p(\J^{candidate})$\;
  }
}
\Return $\J^{selected}$
\end{algorithm}
\vspace{10pt}

Similar to the PTAS for (1,2)-Flowshop-Packing, 
our PTAS for (m,2)-Flowshop-Packing maintains all candidates, denoted by $\mathcal{C}$, of $\J^{selected}$.
For each candidate of $\J^{profitable}$, we guess a distribution of $\J^{profitable}$ over $\F$. 
Then a critical job is guessed on each flowshop, including the dummy flowshop. 
After solving $\NewMLP$,  
a candidate of $\J^{selected}$ is obtained by 
restricting its optimal basic feasible solution to $\F$ and discarding jobs that are assigned fractionally.
Finally, we select the most profitable candidate that satisfies the limited makespan of 1.
A detailed description of this PTAS is provided in Algorithm \ref{ptas2}.

\begin{theorem}\label{thm2}
Algorithm \ref{ptas2} is a PTAS for the (m,2)-Flowshop-Packing problem if $m$ is a fixed constant.
\end{theorem}

\begin{proof}
We claim that the optimal solution $\J^*$ is contained in the candidate collection $\mathcal{C}$.
The approximation ratio $1-\epsilon$ is implied by Lemma \ref{lemma_mlp}.

$\NewMLP$ can be solved in $ploy(mn)$ via any interior point methods \cite{PW00}, 
as it has $O(nm)$ constraints.
By Lemma \ref{lemma_distribution}, there are at most $O((m/\epsilon)^{1/\epsilon})$ distributions of $\J^{profitable}$.
The number of different combinations of critical jobs on flowshops (including the dummy flowshop) 
is upper bounded by $O(n^m)$ since each critical job has at most $n$ choices.
$\mathcal{C}_{\epsilon}$ has a cardinality of $O(n^{(1+\epsilon)/\epsilon})$.
Therefore, the first for-loop in Algorithm \ref{ptas2} takes $poly(n^{O(1/\epsilon + m)}  (m/\epsilon)^{1/\epsilon})$ time. 
The second for-loop in Algorithm \ref{ptas2} takes $poly(n^{1/\epsilon})$ time,
as $\mathcal{C}$ contains $O(n^{(1+\epsilon)/\epsilon})$ candidates and Johnson's algorithm runs in $O(n\log n)$.

The overall time complexity of Algorithm \ref{ptas2} is $poly(n^{O(1/\epsilon + m)}  (m/\epsilon)^{1/\epsilon})$.
This completes the proof of this Theorem.
\end{proof}

\section{Conclusion}\label{sec_conclusion}

We explore an interesting scheduling problem recently introduced by Chen et al. \cite{CHG21}, i.e.,
the scheduling of parallel two-stage flowshops with makespan constraint (or (m,2)-Flowshop-Packing),
which generalizes the classic Multiple Knapsack problem.
Given a limited makespan requirement, 
the goal is to select a subset of two-stage jobs and 
schedule them on multiple flowshops to achieve the maximum profit. 
All existing approximation algorithms \cite{DGR06,CHG21} rely on the connection 
between the Multiple Knapsack problem and the (m, 2)-Flowshop-Packing problem.
To design a PTAS for the case in which the number of flowshops is a fixed constant,
\WT{we utilize the structure of the job sequence under Johnson's order,
and combine guessing techniques and rounding techniques in linear programming.}
%
%
Our PTAS achieves the best possible approximability result 
as the special case with two flowshops \WT{does not admit FPTASes}.
Meanwhile, our PTAS gives a firm answer to an open question presented by Chen et al. \cite{CHG21}
which asks 
whether there exists an PTAS for the case where the number $m$ of flowshops is a fixed constant.

There are several open questions that worth further exploration. 
\begin{enumerate}
\item Our current PTAS has a large time complexity 
        due to exhaustively guessing the substructure of the optimal solution. 
        Is it possible to design an EPTAS? 
\item 
        %
        \WT{The case with a single flowshop is only weakly NP-hard. 
        It would be interesting to investigate whether it admits an FPTAS. }
\item When $m$ is part of the input, Chen et al. \cite{CHG21} presented a $(1/3-\epsilon)$-approximation algorithm. 
        It is unknown whether this case is APX-hard. 
        If not, is it possible to design a PTAS or at least an approximation algorithm with ratio better than $1/3-\epsilon$?
\end{enumerate}

\section*{Acknowledgments}
\WT{
HZ was partially supported by the National Nature Science Foundation of China (Grant No.72071157, No.71732006, No.72192834),
and by the China Postdoctoral Science Foundation (Grant No. 2016M592811).
}

\bibliographystyle{plain}
\bibliography{manuscript-arxiv}

\newpage

\appendix


\end{document}